\documentclass[conference]{IEEEtran}
\IEEEoverridecommandlockouts
\usepackage{graphicx}
\usepackage{array}
\usepackage{amsmath}
\usepackage{float}
\usepackage{multirow}
\usepackage{array}
\usepackage{cancel}

\usepackage{tabu}
\usepackage{amsthm}
\usepackage{epsfig}
\usepackage{epstopdf}
\usepackage{epsf}
\usepackage{color}
\usepackage[english]{babel}

\theoremstyle{definition}
\newtheorem{theorem}{Theorem}
\newtheorem{lemma}[theorem]{Lemma}

\usepackage{amssymb}
\usepackage{cite}
\usepackage{amsmath,amssymb,amsfonts}
\usepackage{graphicx}
\usepackage{textcomp}
\usepackage{xcolor}
\usepackage{algorithm}
\usepackage[noend]{algpseudocode}
\usepackage[T1]{fontenc}
\usepackage{graphicx,lipsum,afterpage,subcaption}
\usepackage{enumitem}
\usepackage{lipsum}
\usepackage{mathtools}
\usepackage{cuted}

\DeclareMathOperator*{\argmin}{argmin}
\begin{document}

	\title
	{Entropy-Constrained Maximizing Mutual Information Quantization}

	\author{\IEEEauthorblockN{Thuan Nguyen}
		\IEEEauthorblockA{School of Electrical and\\Computer Engineering\\
			Oregon State University\\
			Corvallis, OR, 97331\\
			Email: nguyeth9@oregonstate.edu}
		\and
		\IEEEauthorblockN{Thinh Nguyen}
		\IEEEauthorblockA{School of Electrical and\\Computer Engineering\\
			Oregon State University\\
			Corvallis, 97331 \\
			Email: thinhq@eecs.oregonstate.edu}
	}

	\maketitle
	\begin{abstract}
In this paper, we investigate the quantization of the output of a binary input discrete memoryless channel that maximizing the mutual information between the input and the quantized output under an entropy-constrained of the quantized output.  A polynomial time algorithm is introduced that can find the truly global optimal quantizer. This results hold for binary input channels with arbitrary number of quantized output. Finally, we extend this results to binary input continuous output channels and show a sufficient condition such that  a single threshold quantizer is an optimal quantizer. Both theoretical results and numerical results are provided to justify our techniques. 
	\end{abstract}
	Keyword: vector quantization, partition, impurity, concave, constraints, mutual information.
	\section{Introduction}
Recently, the problem of quantization that maximizing the mutual information between input and quantized output is a hot topic in information theory society. The design of that quantizers is important in the sense of designing the communication decoder i.e., polar code decoder \cite{tal2011construct} and LDPC code decoder \cite{romero2015decoding}. Over a past decade, many algorithms was proposed  \cite{kurkoski2014quantization}, \cite{zhang2016low}, \cite{winkelbauer2013channel}, \cite{iwata2014quantizer}, \cite{mathar2013threshold}, \cite{sakai2014suboptimal},  \cite{koch2013low}, \cite{nguyen2018capacities}, \cite{he2019dynamic}. However, due to the non-linear of partition, finding the global optimal of partition $M$ data points to $K$ subsets is difficult in a general setting. Of course, a naive
exhaustive search on the $M$ points results in the time complexity
of $O(K^M)$ which can quickly become computationally
intractable even for modest values of $M$ and $K$. In \cite{zhang2016low}, a iteration algorithm is proposed to find the locally optimal quantizer with time complexity of $O(TKMN)$ where $T$ is the number of iterations in the algorithms and $N$ is the size of the channel input. Unfortunately, these algorithms can get stuck at a locally optimal solution which can be far away from the globally optimal solutions. However, under a special condition of binary input channel $N=2$, the global optimal quantizer can be found efficiently with the polynomial time complexity of $O(M^3)$ in the worst case \cite{kurkoski2014quantization}. The complexity is further reduced to $O(KNM)$ using the famous SMAWK algorithm \cite{iwata2014quantizer}.

As an extension, quantization that maximizes the mutual information  under the entropy-constrained is very important in the sense of limited communication channels. For example, one wants to quantize/compress the data to an intermediate quantized output before transmits these quantized output to the destination over a limited rate communication channel, then the entropy of quantized output that denotes the lowest compression rate is important. Of course, we want to keep the largest mutual information between input and quantized output while the transmission rate is lower than the channel capacity. That said, the problem of quantization that maximizing mutual information under entropy-constrained is an interesting problem that can be applied in many scenarios.
While the problem of quantization that maximizing the mutual information was thoroughly investigated, there is a little of literature about the problem of quantization maximizing mutual information under the entropy-constrained. As the first article,  Strouse et al.   proposed an iteration algorithm to find the local optimal quantizer that maximizing the mutual information under the entropy-constrained of quantized-output \cite{strouse2017deterministic}. In \cite{nguyen2019minimizing}, the authors generalized the results in \cite{strouse2017deterministic} to find the local optimal quantizer that minimizes an arbitrary impurity function while the quantized output constraint is an arbitrary concave function. However, as the best of our knowledge, there is no work that can determine the globally optimal quantizer that maximizes the mutual information under the entropy-constrained even for the binary input channels. It is worth noting that the very similar setting was established long time ago called entropy-constrained scalar quantization \cite{Marco2004PerformanceOL}, \cite{Gyorgy2001OnTS} and entropy-constrained vector quantization \cite{Chou1989EntropyconstrainedVQ}, \cite{Gersho1991VectorQA}, \cite{Zhao2008OnEV} where the quantized output has to satisfy the entropy-constrained and squared-error distortion between input and quantized output is minimized.

 In this paper, we introduce a polynomial time algorithm that can find the truly global optimal quantizer  if the channel input is binary. This result holds for any binary input channel with arbitrary number of quantized output. Finally, we extend this results to binary input continuous binary output channels and show a sufficient condition such that  a single threshold quantizer is an optimal quantizer.
The outline of our paper is as follows.  In Section \ref{sec: formulation}, we describe the problem formulation and its applications.  In Section \ref{sec: solution}, we review the results in learning theory which can be applied to find the optimal quantizer.  In Section \ref{sec: algorithm}, we provide  a polynomial time algorithm to find the truly global optimal quantizer if the channel input is binary. Moreover, we extend the results to binary input continuous binary output channels and state a sufficient condition such that a single threshold quantizer is optimal. The simulation result is provided in Sec. \ref{sec: simulation}. Finally,  we provide a few concluding remarks in Section \ref{sec: conclusion}.

\section{Problem Formulation}
	\label{sec: formulation}
		\begin{figure}
		\centering
		\includegraphics[width=2.5 in]{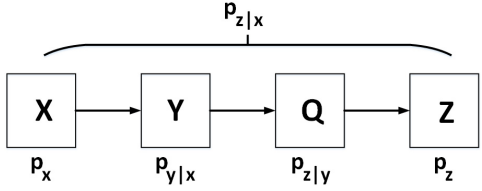}\\
		\caption{A discrete memoryless channel having $N$ inputs and $K$ quantized  outputs using  $K-1$ thresholds.}\label{fig: 1}
	\end{figure}	
	

Fig. \ref{fig: 1} illustrates our channel.   The discrete input $X=\{x_1,x_2,\dots,x_N \}$  with a given pmf $p_X=\{p_{x_1},p_{x_2},\dots,p_{x_N}\}$. Let the channel output $Y=\{y_1,y_2,\dots,y_M \}$ that is specified by the distribution $p_Y=\{p_{y_1},p_{y_2},\dots,p_{y_M}\}$ and  a conditional distribution $p_{y_j|x_i}$ for $j=1,2,\dots,M$, $i=1,2,\dots,N$. The joint distribution $p_{x_n,y_m}$, therefore, is given. The output $Y$ is then quantized to $Z=\{z_1,z_2,\dots,z_K \}$ with the distribution $p_Z=\{p_{z_1},p_{z_2},\dots,p_{z_K} \}$, $K<M$, by a possible stochastic quantizer $Q_{z|y}=p_{z_k|y_j}$. 
One wants to design the optimal quantizer to maximize the mutual information between input $X$ and quantized output $Z$ while the quantized output $Z$ has to satisfy an entropy-constrained. Thus, in this paper we are interested to find the optimal quantizer $Q^*$ such that.

\begin{equation}
\label{eq: problem formulation 1}
Q^*=\max_{Q} \beta I (X;Z)-H(Z).
\end{equation}
where $\beta>0$ is a parameter that controls the trade off between maximizing the mutual information $I(X;Z)$ or minimize the entropy of quantized output $H(Z)$. The mutual information $I(X;Z)$ is defined by
\begin{equation*}
I(X;Z)=H(Z)-H(Z|X)=H(X)-H(X|Z).
\end{equation*}
The entropy function $H(Z)$ is defined by
\begin{equation*}
H(Z)=-\sum_{k=1}^{K}p_{z_k} \log (p_{z_k}).
\end{equation*}
From $p_X$ is given, the problem in (\ref{eq: problem formulation 1}) is equivalent to
\begin{equation}
\label{eq: problem formulation 2}
Q^*=\min_{Q} \beta H (X|Z) + H(Z).
\end{equation}

\section{Connection to minimum impurity under concave constraint problem}
\label{sec: solution}

In this section, we want to establish the connection between the problem of discrete channel quantization maximizing mutual information under the entropy-constrained and the area of statistical learning theory. Similar to the setting in Sec. \ref{sec: formulation}, consider a discrete random variable $X=\{x_1,x_2,\dots,x_N\}$ which is stochastically linked to an observation discrete data $Y=\{y_1,y_2,\dots,y_M\}$. One wants to quantize $Y$ to a smaller levels of quantized subsets $Z=\{z_1,z_2,\dots,z_K\}$ such that the cost function $F(X,Z)$ between $X$ and $Z$ is minimized while the probability of quantized output satisfies a constraint $G(Z) \leq D$ for a pre-specified constant $D$ \cite{nguyen2019minimizing}. 
\begin{equation}
\label{eq: problem formulation 3}
Q^*=\min_{Q} \beta F (X,Z) + G(Z).
\end{equation}

The mapping $Q(Y) \rightarrow Z$ called a classifier/partition in the context of learning theory that is very similar to quantizer in this paper. The cost function $F(.)$ is called the impurity function that is a way to measure the goodness of quantization/classification/partition.
The original impurity $F(X,Z)$ between $X$ and $Z$ is defined by adding up the weighted loss function in each output subset $z_k \in Z$ \cite{burshtein1992minimum}, \cite{chou1991optimal}, \cite{coppersmith1999partitioning}. 
\begin{eqnarray}
\label{eq: definition of impurity}
F(X,Z) &=& \sum_{k=1}^{K} p_{z_k} f[p_{x|z_k}] \nonumber\\
&=&  p_{z_k} f[p_{x_1|z_k},\! p_{x_2|z_k},\! \dots,\! p_{x_N|z_k}]
\label{eq: definition of impurity}
\end{eqnarray}
where $p_{x|z_k}=[p_{x_1|z_k}, p_{x_2|z_k}, \dots, p_{x_N|z_k}]$ denotes the conditional distribution $p_{x|z_k}$. The factor $p_{z_k}$ denotes the weight of subset $z_k$. The function $f(.)$ is an arbitrary \textit{concave function}. 

The impurity function in (\ref{eq: definition of impurity}) can be rewritten by the function of the joint distribution $p_{x_n,z_k}$ for $n=1,2,\dots,N$ and $k=1,2,\dots,K$. Let define 
\begin{equation*}
p_{x_n,z_k}=\sum_{y_j \in Y}^{}p_{x_n,y_j}p_{z_k|y_j}. 
\end{equation*}
Thus,
\begin{small}
\begin{eqnarray}
p_{z_k} f[p_{x|z_k}] &\!=\!&  (\sum_{n\!=\!1}^{N}p_{x_n,z_k}) f[\dfrac{p_{x_1\!,\!z_k}}{\sum_{n\!=\!1}^{N}p_{x_n\!,\!z_k}}, \dots, \dfrac{p_{x_N\!,\!z_k}}{\sum_{n\!=\!1}^{N}p_{x_n\!,\!z_k}}] \nonumber\\
\label{eq: new formulation}
\end{eqnarray}
\end{small}
where $\sum_{n=1}^{N}p_{x_n,z_k}$ denotes the weight of $z_k$ and   $p_{x|z_k}=[\dfrac{p_{x_1,z_k}}{\sum_{n=1}^{N}p_{x_n,z_k}}, \dots, \dfrac{p_{x_N,z_k}}{\sum_{n=1}^{N}p_{x_n,z_k}}] $ denotes conditional distribution $p_{x|z_k}$.  \textit{The impurity function, therefore, is only the function of the joint distribution $p_{x_n,z_k}$}. 

Various of common impurity functions have been suggested in \cite{burshtein1992minimum}. However, in this paper, we are interested to the entropy impurity function such that
\begin{equation*}
f[p_{x|z_k}] = - \sum_{n=1}^{N} p_{x_n|z_k} \log (p_{x_n|z_k}).
\end{equation*}
Thus,
\begin{small}
\begin{eqnarray}
F(X,Z)&=&\sum_{k=1}^{K} p_{z_k}[  \sum_{n=1}^{N} - p_{x_n|z_k} \log p_{x_n|z_k} ] \nonumber\\
&=&\sum_{i=1}^{K} p_{z_k} H(X|Z_k) = H(X|Z) \label{eq: 5}.
\end{eqnarray}
\end{small}
On the other hand, the constraint $G(Z)$ can be an arbitrary \textit{concave function} over the quantized output $p_Z$. Thus, the entropy-constrained can be constructed if $G(.)$ is entropy function. 
\begin{small}
\begin{equation}
\label{eq: 6}
G(Z)= -\sum_{k=1}^{K}p_{z_k} \log (p_{z_k}).
\end{equation}

\end{small}
From (\ref{eq: 5}), (\ref{eq: 6}), obviously that \textit{ the problem in (\ref{eq: problem formulation 2}) is a sub-problem of the problem in (\ref{eq: problem formulation 3})}. That said, all the elegant and general theoretical results in \cite{nguyen2019minimizing} can be applied to solve problem (\ref{eq: problem formulation 2}). 
Based on the general results in \cite{nguyen2019minimizing}, the optimal quantizer has the followings properties: (i) the optimal quantizer is a deterministic quantizer, i.e., the partition is a hard partition and therefore $p_{z_k|y_i}=0$ or $p_{z_k|y_i}=1$ for $\forall$ $i,k$; (ii) the optimal quantizer is equivalent to hyper-plane cuts in the space of posterior distribution; (iii) the necessary optimality condition for optimal quantizer is established. The detail results are followings. 

\subsection{Hard partition is optimal }
Noting that in the setting of problem (\ref{eq: problem formulation 2}), the optimal quantizer may be a stochastic (soft) quantizer  i.e., $0 \leq p_{z_k|y_j} \leq 1$ or a data $y_i$ can belong to more than a quantized output with an arbitrary probability.  However, from the Lemma 1 in \cite{nguyen2019minimizing}, the optimal quantizer is a hard quantizer. That said, each data $y_i$ is quantized to a deterministic output $z_l$ or $p_{z_k|y_j} =\{0,1\}$. Thus, we can reduce our interest to only the deterministic quantizers. Lemma 1 extends the ideas in \cite{kurkoski2014quantization}, showing that the purely stochastic quantizers, that is, nondeterministic quantizers, never have better performance than deterministic quantizers under an entropy-constrained quantization.

\subsection{Necessary optimality condition for an optimal quantizer}

From Theorem in \cite{nguyen2019minimizing}, we note that the optimal quantizer $Q^*$ should allocate the data $y_i$ to $z_l$ if $D(y_i,z_l) < D(y_i,z_k)$ $\forall$ $k=1,2,\dots,K$ and $k \neq l$ where $D(y_i,z_l)$ is the "distance" from $y_i$ to $z_l$. This result is stated as the following.

\begin{theorem}
\label{theorem: 1}
Suppose that  an optimal partition $Q^*$ yields the optimal output $Z=\{z_1,z_2,\dots,z_K \}$.
 For each optimal set $z_l$, $l \in \{1,2,\dots,K\}$, we define vector $c_l=[c_l^1,c_l^2,\dots,c_l^N]$:
\begin{equation}
\label{eq: 16}
c_l^n= \frac{\partial p_{z_l}f[p_{x|z_l}] }{\partial p_{x_n,z_l}}, \forall n \in \{1,2,\dots,N\}, 
\end{equation}
where $p_{z_l}f[p_{x|z_l}$ is defined in (\ref{eq: new formulation}). 
We also define 
\begin{equation}
\label{eq: 17}
d_l=\frac{\partial G(Z)}{\partial p_{z_l}}.
\end{equation}

Define the "distance" from data $y_i \in Y$ to data $z_l$ is 
\begin{eqnarray}
D(y_i,z_l) = \beta \sum_{n=1}^{N}[p_{x_n|y_i} c_l^n] + d_l \label{eq: optimality condition}.
\end{eqnarray}

Then, data $y_i$ is quantized to $z_l$ if and only if $D(y_i,z_l) \leq D(y_i,z_k)$ for $\forall k \in \{1,2,\dots,K, k \neq l\}$. 
\end{theorem}
\begin{proof}
Please see Theorem 1 in \cite{nguyen2019minimizing}. 
\end{proof}

Since (\ref{eq: problem formulation 2}) is a sub-problem of (\ref{eq: problem formulation 3}) using $f[.]$ is entropy function and the constraint $G(.)$ is entropy-constrained, the distance $D(y_i,z_l)$ is specified in the following lemma.  
\begin{lemma}
\label{lemma: 1}
The optimal quantizer $Q^*=\min_{Q} \beta H (X|Z) + H(Z)$ should quantize the data $y_i$ to $z_l$ if and only if $D(y_i,z_l) \leq D(y_i,z_k)$ for $\forall k \in \{1,2,\dots,K, k \neq l\}$ where
\begin{eqnarray}
D(y_i,z_l) &=& \beta \sum_{n=1}^{N}p_{x_n|y_i} \log (\dfrac{p_{x_n|y_i}}{p_{x_n|z_l}}) - \log (p_{z_l}). 
\label{eq: distance this problem}
\end{eqnarray}
\end{lemma}
\begin{proof}
By taking the derivative of $c_l^n$ and $d_l$ in Theorem \ref{theorem: 1} and ignoring the constant without changing the different between $D(y_i,z_l)-D(y_i,z_k)$, $k \neq l$, the new distance metric is constructed. 
\end{proof}

The first component in the distance $D(y_i,z_l)$ is actually the Kullback-Leibler distance between data and the centroid of quantized output \cite{zhang2016low} while the second component denotes the impact of entropy-constrained. Obviously that minimizing $-log (p_{z_l})$ meaning that one should quantize $y_i$ to a quantized output $z_l$ that already having a large probability.

\subsection{Separating hyper-plane condition for optimality}
Let $p_{x|y_j}$ be the conditional probability distribution that is defined by a vector in $N$ dimensional probability space
\begin{equation*}
p_{x|y_j}=[p_{x_1|y_j},p_{x_2|y_j},\dots,p_{x_{N}|y_j}]
\end{equation*}
where $0 \leq p_{x_i|y_j} \leq 1$ and $\sum_{i=1}^{N}p_{x_i|y_j} = 1$. 
Thus, each data  $y_i$ is equivalent to a $N-1$ dimensional vector $\bar{p_{x|y_j}}$. 
\begin{equation*}
\bar{p_{x|y_j}}=[p_{x_1|y_j},p_{x_2|y_j},\dots,p_{x_{N-1}|y_j}].
\end{equation*}
For convenient, we denote $\bar{p_{x|y_j}}=v_j$. Now, the quantizer $Q(Y) \rightarrow Z$ is equivalent to a quantizer $\bar{Q}$
\begin{equation*}
\bar{Q} : \{v_1,v_2,\dots,v_M\} \rightarrow \{z_1,z_2,\dots,z_K\}.
\end{equation*}
Noting that two quantizers $Q$ and $\bar{Q}$ are equivalent in the sense that if $Q(y_j)=z_k$ then $\bar{Q}(v_j)=z_k$. From Sec. III-C in \cite{nguyen2019minimizing}, the following Theorem holds.
\begin{theorem}
\label{theorem: 2}
There exist an optimal quantizer $Q^*$ such that the optimal partition is separated by hyper-plane cuts in the space of posterior probability $\bar{p_{x|y_j}}=v_j$.
\end{theorem}

The advantage of using quantizer $\bar{Q}$ is that it works in $N-1$ dimensional space while $Q$ works in the $N$ dimensional space.  For a large value of $N$, the different is negligible. However, if $N$ is small, i.e., $N=2$, the Theorem \ref{theorem: 2} gives a powerful condition which is characterized more detail in Sec. \ref{sec: algorithm}. 

\section{Quantizer Design Algorithm}
 \label{sec: algorithm}

\subsection{Quantizer algorithm for binary input discrete output channels}
In this section, we consider the channels with binary input distribution i.e., $|X|=N=2$. From Theorem \ref{theorem: 2}, the optimal quantizer is equivalent to a hyper-plane in one dimensional space. However, a hyper-plane in one dimensional space is a point and the optimal quantizer is equivalent to scalar quantizer in the order of $v_{j}=\bar {p_{x|y_i}}=p_{x_1|y_j}$. That said, existing $K+1$ thresholds $a=(a_0=0,a_1,\dots,a_{K-1},a_K=1)$ such that
\begin{equation*}
a_0<a_1<\dots < a_{K-1}<a_K
\end{equation*}

and 
\begin{equation*}
Q^*(v_j)=z_k, \text{  if  }  a_{k-1} <p_{x_1|y_j} < a_{k}.
\end{equation*}

Without the loss of generality, we can order the data set $Y$ by the descending order of $p_{x_1|y_j}$ in the time complexity of $O(M \log M)$. Thus, in the rest part of this paper, we suppose that
\begin{equation}
p_{x_1|y_1} \leq p_{x_1|y_2} \leq \dots \leq p_{x_1|y_{M-1}} \leq p_{x_1|y_M}.
\end{equation}

The optimal quantizer, therefore, can be found by searching the optimal scalar thresholds $0=a_0^*<a_1^*<\dots <a_{K-1}^*<a_K^*=1$. Therefore, the problem can be cast as a 1-dimensional quantization/clustering problem that can be solved efficiently using the famous dynamic programming \cite{kurkoski2014quantization}, \cite{wang2011ckmeans}. The detail algorithm is proposed in Algorithm \ref{alg: 2}. 

\begin{footnotesize}
	\begin{algorithm}
		\caption{Dynamic programming for finding $D(1,M,K)$}
		\label{alg: 2}
		\begin{algorithmic}[1]
			\State{\textbf{Input}: $p_X$, $p_Y$, $p_{Y|X}$, $M$, $K$.}
			
			\State{\textbf{Initialization}: $D(i,j,k)=0$ for $\forall$ $j=0$ or $k=0$. }
			
			\State{\textbf{Recursion step}: }
			
			\State{\hspace{.1 in} For $k=1,2,\dots,K$}
			
			\State{\hspace{.3 in} For $j=1,2,\dots,M$}
\begin{equation*}
D(i,j,k)=\min_{0 \leq q \leq j-1} \{ D(i,q,k-1)+D(q+1,j,1) \}.
\end{equation*}
			\State{\hspace{.3 in} End For}	
			\State{\hspace{.3 in} Store the local decision:}
\begin{equation*}
H_k(j)=\argmin_{q} \{ D(i,q,k-1)+D(q+1,j,1) \}.
\end{equation*}						
			\State{\hspace{.1 in} End For}	
			\State{\textbf{Backtracking step}: Let $a_K^*=M$, for each $i=\{K-1,K-2,\dots,1\}$}
			\begin{equation*}
			a_i^*=H_{i+1}(a_{i+1}^*).
			\end{equation*}

			\State{\textbf{Output}: $D(1,M,K)$, $a^*=\{a_0, a_1^*,\dots,a_{K-1}^*,a_K^* \}$.}
		\end{algorithmic}
	\end{algorithm}
\end{footnotesize}
Now, let us define $D(i,j,k)$ as the minimum (optimal) value of $\beta H(X|Z) + H(Z)$ by partition $(y_i,y_j]$ into $k$ subsets where $0 \leq i \leq j \leq M$ and $0 \leq k \leq K$.  Each $D(i,j,k)$ is the result of using an optimal quantizer $Q^*(i,j,k)$ which separates the data in $(y_i,y_j]$ to $k$ clusters using $k-1$ thresholds. For a given $Q^*(i,j,k)$, define $w(i,j,k)$ and $t(i,j,k)$ as the values of the conditional entropy $\beta H(X|Z)$ and entropy $H(Z)$ associated with the optimal quantizer $Q^*(i,j,k)$, then
 \begin{equation}
 \label{eq: total cost function}
 D(i,j,k)=w(i,j,k)+ t(i,j,k).
 \end{equation} 
Now, the key of the dynamic programming algorithm is based on the following recursion:
 \begin{equation}
 \label{eq: recursive}
 D(i,j,k)=\min_{0 \leq q \leq j-1} \{ D(i,q,k-1)+D(q+1,j,1) \}.
 \end{equation}
In the above recursion, the value of $k$ partitions with a total of $j$ elements can be written as the sum of $k-1$ partitions with $q$ elements and one additional partition with $j-q$ elements.  Thus, minimum value can be found by searching for the right index $q$, and the recursion follows.   Again, we note that this dynamic programming approach works because the value of the large partition equals  the sum of the values of its smaller sub-partitions.
 
Now, consider initial values $D(i,j,k)=0$ if $j=0$ or $k=0$. From this initial values, using (\ref{eq: recursive}), one can compute all of $D(i,j,k)$. The optimal solution is $D(1,M,K)$. After finding the optimal solution, one can use the backtracking method to find all the optimal thresholds. The backtracking step is performed by storing the indices that result in the minimum values.  Specifically,
 \begin{equation}
 \label{eq: store local decision}
 H_k(j)=\argmin_{q} \{ D(i,q,k-1)+D(q+1,j,1) \}.
 \end{equation}	
 Then, $H_k (j)$ saves the position of ${k-1}^{th}$ threshold. Finally, let  $a_K^*=M$, for each $i= \{ K-1,K-2,\dots,1 \} $, all of other optimal thresholds can be found by backtracking. 
 \begin{equation}
 \label{eq: backtracking}
 a_i^*=H_{i+1}(a_{i+1}^*).
 \end{equation}				
\textbf{ Complexity.}  Noting that except step 5 in Algorithm \ref{alg: 2} takes the time complexity of $O(KM^2)$, other steps can be done in a linear time. Thus, the total time complexity of Algorithm \ref{alg: 2} is $O(KM^2)$.

\subsection{Quantization for binary input binary output continuous channels}

In this section, we extend the previous results to the discrete binary input continuous binary output channels. Consider a channel with discrete input $X=\{x_1,x_2 \}$ which is corrupted by a noise having  continuous distribution to produce the continuous output $y \in Y=\mathbf{R}$. Thus, the channel is specified by two continuous distribution $p_{y|x_1}=\phi_1(y)$ and $p_{y|x_2}=\phi_2(y)$. One wants to quantize continuous output $y$ back to the discrete binary quantized output $Z=\{z_1,z_2\}$. Now, consider the following variable $r(y)$
\begin{equation*}
r(y)=p_{x_1|y}=\dfrac{p_1 p_{y|x_1}}{p_1 p_{y|x_1} + p_2 p_{y|x_2}  }=\dfrac{p_1 \phi_1(y)}{p_1 \phi_1(y) + p_2 \phi_2(y)}.
\end{equation*}
Since $N=2$ and $K=2$, from the result in Theorem \ref{theorem: 2}, the optimal quantizer can be found by searching an optimal scalar threshold $0 <a^* < 1$ such that 
$$
\begin{cases}
Q(y)=z_1 \text{   if } r(y) \leq a^*,\\
Q(y)=z_2 \text{   if } r(y) > a^*.
\end{cases}
$$
Thus, the optimal quantizer can be found by an exhausted searching  over a new random variable  $0 <a < 1$. The complexity of this algorithm is $O(M)$ where $M=\dfrac{1}{\epsilon}$ and $\epsilon$ is a small number denotes the precise of the solution. From the optimal value $a^*$, the corresponding thresholds $y \in Y$ can be constructed. Interestingly, the following Lemma shows a sufficient condition where a single threshold $y \in Y=R$ is an optimal quantizer.  
\begin{lemma}
\label{lemma: 1}
If $\dfrac{\phi_2(y)}{\phi_1(y)}$ is a strictly increasing/decreasing function, a single threshold quantizer is optimal.
\end{lemma}

\begin{proof}
We consider
\begin{equation*}
r(y)=p_{x_1|y}=\dfrac{p_1 \phi_1(y)}{p_1 \phi_1(y) + p_2 \phi_2(y)}=\dfrac{1}{1+\dfrac{\phi_2(y)}{\phi_1(y)}}.
\end{equation*}
Since $\dfrac{\phi_2(y)}{\phi_1(y)}$ is a strictly increasing/decreasing function, $r(y)$ is a strictly increasing/decreasing function. Thus, for a given value of $a$, existing a single value of $y$ such that $r(y)=a$. Therefore, the optimal $a^*$ corresponds to a single value of $y^*$. Thus, a single threshold quantizer is optimal in this context. Our result is an extension of Lemma 2 in \cite{kurkoski2017single}. 
\end{proof}
\section{Numerical results}
 \label{sec: simulation}

 Consider a communication system which transmits input $X=(x_1=-2,x_2=2)$ having $p_{x_1}=p_{x_2}=0.5$ over an additive noise channel with i.i.d Gaussian noise $N(\mu=0,\sigma=1)$.  Due to the additive property, the conditional density of output $y \in Y=R$ given input $x_1$ is $\phi_1(y)=p_{y|x_1}=N(-2,1)$ while the conditional density of output $y$ given input $x_2$ is $\phi_2(y)=p_{y|x_2}=N(2,1)$.  The continuous output $y$ then is quantized to $N=4$ output levels $Z=(z_1,z_2,z_3,z_4)$.  We first discrete $y \in Y=R$ to $M=200$ pieces from $[-10,10]$ with the same width $\epsilon=0.1$. Thus, $Y=\{y_1,y_2,\dots,y_{200}\}$ with the conditional density $p_{y_j|x_i}$ and $p_{y_j}$, $\forall$ $i=1,2$ and $j=1,2,\dots,200$ can be determined using two given conditional densities $\phi_1(y)=N(-2,1)$ and $\phi_2(y)=N(2,1)$. For $\beta=2,\dots,13$, the curve in Fig. \ref{fig: 2} illustrates the optimal pairs $(I^*(X,Z),H^*(Z))$. For example, if one requires that $H(Z) \leq 1.18$, we should pick $\beta=6$ that produces $H^*(Z) =1.795$ and $I^*(X,Z)=0.87274$. 
		\begin{figure}
		\centering
		\includegraphics[width=2.7 in]{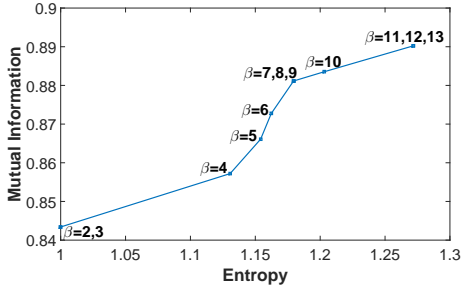}\\
		\caption{Optimal pairs $(I^*(X,Z),H^*(Z))$ corespond to $\beta=1,2,\dots,13$.}\label{fig: 2}
	\end{figure}	 
 
 
\section{Conclusion}
 \label{sec: conclusion} 
A polynomial time complexity algorithm is proposed that can find the globally optimal quantizer to maximize the mutual information between the input and the quantized output under an entropy-constrained if the channel input is binary. This result holds for any binary input channels with arbitrary number of quantized output. We also extend the result to binary input continuous binary output channels and show a sufficient condition such that  a single threshold quantizer is optimal. Both theoretical results and numerical results are provided to justify our techniques. 
\
\bibliographystyle{unsrt}
\bibliography{sample}

\begin{thebibliography}{10}

\bibitem{tal2011construct}
Ido Tal and Alexander Vardy.
\newblock How to construct polar codes.
\newblock {\em arXiv preprint arXiv:1105.6164}, 2011.

\bibitem{romero2015decoding}
Francisco Javier~Cuadros Romero and Brian~M Kurkoski.
\newblock Decoding ldpc codes with mutual information-maximizing lookup tables.
\newblock In {\em Information Theory (ISIT), 2015 IEEE International Symposium
  on}, pages 426--430. IEEE, 2015.

\bibitem{kurkoski2014quantization}
Brian~M Kurkoski and Hideki Yagi.
\newblock Quantization of binary-input discrete memoryless channels.
\newblock {\em IEEE Transactions on Information Theory}, 60(8):4544--4552,
  2014.

\bibitem{zhang2016low}
Jiuyang~Alan Zhang and Brian~M Kurkoski.
\newblock Low-complexity quantization of discrete memoryless channels.
\newblock In {\em 2016 International Symposium on Information Theory and Its
  Applications (ISITA)}, pages 448--452. IEEE, 2016.

\bibitem{winkelbauer2013channel}
Andreas Winkelbauer, Gerald Matz, and Andreas Burg.
\newblock Channel-optimized vector quantization with mutual information as
  fidelity criterion.
\newblock In {\em 2013 Asilomar Conference on Signals, Systems and Computers},
  pages 851--855. IEEE, 2013.

\bibitem{iwata2014quantizer}
Ken-ichi Iwata and Shin-ya Ozawa.
\newblock Quantizer design for outputs of binary-input discrete memoryless
  channels using smawk algorithm.
\newblock In {\em Information Theory (ISIT), 2014 IEEE International Symposium
  on}, pages 191--195. IEEE, 2014.

\bibitem{mathar2013threshold}
Rudolf Mathar and Meik D{\"o}rpinghaus.
\newblock Threshold optimization for capacity-achieving discrete input one-bit
  output quantization.
\newblock In {\em Information Theory Proceedings (ISIT), 2013 IEEE
  International Symposium on}, pages 1999--2003. IEEE, 2013.

\bibitem{sakai2014suboptimal}
Yuta Sakai and Ken-ichi Iwata.
\newblock Suboptimal quantizer design for outputs of discrete memoryless
  channels with a finite-input alphabet.
\newblock In {\em Information Theory and its Applications (ISITA), 2014
  International Symposium on}, pages 120--124. IEEE, 2014.

\bibitem{koch2013low}
Tobias Koch and Amos Lapidoth.
\newblock At low snr, asymmetric quantizers are better.
\newblock {\em IEEE Trans. Information Theory}, 59(9):5421--5445, 2013.

\bibitem{nguyen2018capacities}
Thuan Nguyen, Yu-Jung Chu, and Thinh Nguyen.
\newblock On the capacities of discrete memoryless thresholding channels.
\newblock In {\em 2018 IEEE 87th Vehicular Technology Conference (VTC Spring)},
  pages 1--5. IEEE, 2018.

\bibitem{he2019dynamic}
Xuan He, Kui Cai, Wentu Song, and Zhen Mei.
\newblock Dynamic programming for discrete memoryless channel quantization.
\newblock {\em arXiv preprint arXiv:1901.01659}, 2019.

\bibitem{strouse2017deterministic}
DJ~Strouse and David~J Schwab.
\newblock The deterministic information bottleneck.
\newblock {\em Neural computation}, 29(6):1611--1630, 2017.

\bibitem{nguyen2019minimizing}
Thuan Nguyen and Thinh Nguyen.
\newblock Minimizing impurity partition under constraints.
\newblock {\em arXiv preprint arXiv:1912.13141}, 2019.

\bibitem{Marco2004PerformanceOL}
Daniel Marco and David~L. Neuhoff.
\newblock Performance of low rate entropy-constrained scalar quantizers.
\newblock {\em International Symposium onInformation Theory, 2004. ISIT 2004.
  Proceedings.}, pages 495--, 2004.

\bibitem{Gyorgy2001OnTS}
A.~Gyorgy and Tam{\'a}s Linder.
\newblock On the structure of entropy-constrained scalar quantizers.
\newblock {\em Proceedings. 2001 IEEE International Symposium on Information
  Theory (IEEE Cat. No.01CH37252)}, pages 29--, 2001.

\bibitem{Chou1989EntropyconstrainedVQ}
Philip~A. Chou, Tom~D. Lookabaugh, and Robert~M. Gray.
\newblock Entropy-constrained vector quantization.
\newblock {\em IEEE Trans. Acoustics, Speech, and Signal Processing},
  37:31--42, 1989.

\bibitem{Gersho1991VectorQA}
Allen Gersho and Robert~M. Gray.
\newblock Vector quantization and signal compression.
\newblock In {\em The Kluwer international series in engineering and computer
  science}, 1991.

\bibitem{Zhao2008OnEV}
David~Yuheng Zhao, Jonas Samuelsson, and Mattias Nilsson.
\newblock On entropy-constrained vector quantization using.
\newblock 2008.

\bibitem{burshtein1992minimum}
David Burshtein, Vincent Della~Pietra, Dimitri Kanevsky, and Arthur Nadas.
\newblock Minimum impurity partitions.
\newblock {\em The Annals of Statistics}, pages 1637--1646, 1992.

\bibitem{chou1991optimal}
Philip~A. Chou.
\newblock Optimal partitioning for classification and regression trees.
\newblock {\em IEEE Transactions on Pattern Analysis \& Machine Intelligence},
  (4):340--354, 1991.

\bibitem{coppersmith1999partitioning}
Don Coppersmith, Se~June Hong, and Jonathan~RM Hosking.
\newblock Partitioning nominal attributes in decision trees.
\newblock {\em Data Mining and Knowledge Discovery}, 3(2):197--217, 1999.

\bibitem{wang2011ckmeans}
Haizhou Wang and Mingzhou Song.
\newblock Ckmeans. 1d. dp: optimal k-means clustering in one dimension by
  dynamic programming.
\newblock {\em The R journal}, 3(2):29, 2011.

\bibitem{kurkoski2017single}
Brian~M Kurkoski and Hideki Yagi.
\newblock Single-bit quantization of binary-input, continuous-output channels.
\newblock In {\em 2017 IEEE International Symposium on Information Theory
  (ISIT)}, pages 2088--2092. IEEE, 2017.

\end{thebibliography}

\end{document}